\newif\ifabstract

\abstractfalse 
\newif\iffull
\ifabstract \fullfalse \else \fulltrue \fi

\ifabstract
\documentclass[11pt]{article}
\usepackage{fullpage}
\fi

\iffull
\documentclass[11pt]{article}
\usepackage{fullpage}
\fi

\usepackage{amssymb,amsmath}
\usepackage{amsthm}
\usepackage{graphicx}
\usepackage{subfigure}
\usepackage{graphics}
\usepackage{url}
\usepackage{hyperref}
\usepackage{color}
\usepackage{wrapfig,picins}

\usepackage{xspace}
\usepackage{algorithm}
\usepackage{algpseudocode}

\newtheorem{theorem}{Theorem}[section]
\newtheorem{lemma}[theorem]{Lemma}
\newtheorem{proposition}[theorem]{Proposition}
\newtheorem{claim}[theorem]{Claim}
\newtheorem{corollary}[theorem]{Corollary}
\newtheorem{definition}[theorem]{Definition}

\newcommand{\diam}		{\Delta}

\DeclareMathOperator*{\argmin}{argmin}

\renewcommand{\phi}{\varphi}
\newcommand{\eps}{\varepsilon}

\newcommand{\comment}[1]{}

\newcommand{\randbem}[3]{}%

\ifabstract

\fi
\iffull

\fi

\definecolor{darkred}{rgb}{1, 0.1, 0.3}
\definecolor{darkgreen}{rgb}{0.5, 0.8, 0.1}
\definecolor{darkpurple}{rgb}{1.0, 0, 1.0}
\definecolor{darkblue}{rgb}{0, 0, 1.0}

\newcommand{\denselist}{\itemsep 0pt\parsep=1pt\partopsep 0pt}

\newcommand{\X}		{\mathcal{X}}
\newcommand{\Y}		{\mathcal{Y}}
\newcommand{\Z}		{\mathcal{Z}}
\newcommand{\C}		{\mathcal{C}}
\newcommand{\U}		{\mathcal{U}}
\newcommand{\oset}		{K}
\newcommand{\stemv}	{s}

\newcommand{\noembedding}	{outlier embedding\xspace}

\title{Metric embeddings with outliers}
\author{
Anastasios Sidiropoulos\thanks{Dept.~of Computer Science and Engineering and Dept.~of Mathematics, The Ohio State University. Columbus, OH, USA.
Supported by NSF grants CCF 1423230 and CAREER 1453472.}
\and
Yusu Wang\thanks{Dept.~of Computer Science and Engineering, The Ohio State University. Columbus, OH, USA. The work is partially supported by NSF under grant CCF-1319406.}
}
\date{}

\begin{document}

\maketitle

\thispagestyle{empty}
\setcounter{page}{0}

\begin{abstract}
We initiate the study of metric embeddings with \emph{outliers}.
Given some metric space $(X,\rho)$ we wish to find a small set of outlier points $\oset\subset X$ and 
either an isometric or a low-distortion embedding of $(X\setminus \oset,\rho)$ into some target metric space.
This is a natural problem that captures scenarios where a small fraction of points in the input corresponds to noise. 

For the case of isometric embeddings we derive polynomial-time approximation algorithms for minimizing the number of outliers when the target space is an ultrametric, a tree metric, or some constant-dimensional Euclidean space. 
The approximation factors are $3$, $4$ and $2$, respectively. 
For the case of embedding into an ultrametric or tree metric, we further improve the running time to $O(n^2)$ for an $n$-point input metric space, which is optimal.
We complement these upper bounds by showing that outlier embedding into ultrametrics, trees, and $d$-dimensional Euclidean space for any $d\geq 2$ are all NP-hard, as well as NP-hard to approximate within a factor better than 2 assuming the Unique Game Conjecture.

For the case of non-isometries we consider embeddings with small $\ell_{\infty}$ distortion. 
We present polynomial-time \emph{bi-criteria} approximation algorithms.
Specifically, given some $\eps>0$, let $k_\eps$ denote the minimum number of outliers required to obtain an embedding with distortion $\eps$. 
For the case of embedding into ultrametrics we obtain a polynomial-time algorithm which computes a set of at most $3k_{\eps}$ outliers and an embedding of the remaining points into an ultrametric with distortion $O(\eps \log n)$.
Finally, for embedding a metric of unit diameter into constant-dimensional Euclidean space we present a polynomial-time algorithm which computes a set of at most $2k_{\eps}$ outliers and an embedding of the remaining points with distortion $O(\sqrt{\eps})$.
\end{abstract}

\newpage
\setcounter{page}{1}

\section{Introduction}
Metric embeddings provide a framework for addressing in a unified manner a variety of data-analytic tasks.
Let ${\cal X}=(X,\rho)$, ${\cal Y}=(Y,\rho')$ be metric spaces.
At the high level, a metric embedding is a mapping $f:X\to Y$ that either is isometric or preserves the pairwise distances up to some small error called the \emph{distortion}\footnote{Various definitions of distortion have been extensively considered, including multiplicative, additive, average, and $\ell_p$ distortion, as well as expected distortion when the map $f$ is random \cite{bartal1996probabilistic}.}.
The corresponding computational problem is to decide whether an isometry $f$ exists or, more generally, to find a mapping $f$ with minimum distortion.
The space ${\cal Y}$ might either be given or it might be constrained to be a member of a collection of spaces, such as trees, ultrametrics, and so on.
The problems that can be geometrically abstracted using this language include phylogenetic reconstruction (e.g.~via embeddings into trees \cite{agarwala1998approximability,ailon2005fitting,badoiu2007approximation} or ultrametrics \cite{farach1995robust,alon2008ordinal}), visualization (e.g.~via embeddings into constant-dimensional Euclidean space \cite{badoiu2005low,fellows2008parameterized,badoiu2005approximation,matouvsek2010inapproximability,edmonds2010inapproximability,BadoiuCIS06,BergOS13,fellows2009distortion}), and many more (for a more detailed exposition we refer the reader to \cite{indyk20048,indyk2001algorithmic}).

Despite extensive research on the above metric embedding paradigm, essentially nothing is known when the input space ${\cal X}$ can contain \emph{outliers}.
This scenario is of interest for example in applications where outliers can arise from measurement errors.
Another example is when real-world data does not perfectly fit a model due to mathematical simplifications of physical processes.

We propose a generalization of the above high-level metric embedding problem which seeks to address such scenarios:
Given ${\cal X}$ and ${\cal Y}$ we wish to find some small $\oset\subseteq X$ and either an isometric or low-distortion  mapping $f:X\setminus \oset\to Y$.
We refer to the points in $\oset$ as \emph{outliers}.

We remark that it is easy to construct examples of spaces ${\cal X}$, ${\cal Y}$ where any embedding $f:X\to Y$ has arbitrarily large distortion (for any ``reasonable'' notion of distortion), yet there exists $x\in X$ and an isometry $f:X\setminus \{x\}\to Y$.  
Thus new ideas are needed to tackle the more general metric embedding problem in the presence of outliers.

\subsection{Our contribution}

\paragraph{Approximation algorithms.}
We focus on embeddings into ultrametrics, trees, and constant-dimensional Euclidean space. We first consider the problem of computing a minimum size set of outliers such that the remaining point-set admits an isometry into some target space.
We refer to this task as the \emph{minimum outlier embedding problem}.

\emph{Outlier embeddings into ultrametrics.}
It is well-known that a metric space is an ultrametric if and only if any $3$-point subset is an ultrametric.
We may therefore obtain a $3$-approximation as follows:
For all $(x,y,z)\in X^3$, if the triple $(x,y,z)$ is not an ultrametric then remove $x$, $y$, and $z$ from $X$.
It is rather easy to see that this gives a $3$-approximation for the minimum outlier embedding problem into ultrametrics (as for every triple of points that we remove, at least one of them must be an outlier in any optimal solution), with running time $O(n^3)$.
By exploiting further structural properties of ultrametrics, we obtain a $3$-approximation with running time $O(n^2)$.
We remark that this running time is optimal since the input has size $\Theta(n^2)$ and it is straightforward to show that any $3$-approximation has to read all the input (e.g.~even to determine whether ${\cal X}$ is an ultrametric, which corresponds to the case where the minimum number of outliers is zero).

\emph{Outlier embeddings into trees.}
Similarly to the case of ultrametrics, it is known that a space is a tree metric if and only if any $4$-point subset is a tree metric.
This similarly leads to a $4$-approximation algorithm in $O(n^4)$ time. 
We further improve the running time to $O(n^2)$, which is also optimal.
However, obtaining this improvement is significantly more complicated than the case of ultrametrics.

\emph{Outlier embeddings into $\mathbb{R}^d$.}
It is known that for any $d\geq 1$ any metric space admits an isometric embedding into $d$-dimensional Euclidean space if and only if any subset of size $d+3$ does \cite{menger1931new}.
This immediately implies a $(d+3)$-approximation algorithm for outlier embedding into $d$-dimensional Euclidean space  with running time $O(n^{d+3})$, for any $d\geq 1$.
Using additional rigidity properties of Euclidean space we obain a $2$-approximation with the same running time.

\paragraph{Hardness of approximation.}
We show that, assuming the Unique Games Conjecture \cite{khot2002power}, the problems of computing a minimum outlier embedding into ultrametrics, trees, and $d$-dimensional Euclidean space for any $d\geq 2$, are all NP-hard to approximate within a factor of $2-\nu$, for any $\nu>0$.
These inapproximability results are obtained by combining reductions from Vertex Cover to minimum outlier embedding and the known hardness result for the former problem \cite{KR08}.
Note that for the case of embedding into $d$-dimensional Euclidean space for any $d\geq 2$ this inapproximability result matches our upper bound.

\paragraph{Bi-criteria approximation algorithms.}
We also consider non-isometric embeddings.
All our results concern $\ell_{\infty}$ distortion.
For some outlier set $\oset\subseteq X$, the $\ell_{\infty}$ distortion of some map $f:X\setminus \oset\to Y$ is defined to be
\[
\sup_{x,y\in X\setminus \oset} \left|\rho(x,y)-\rho'(f(x),f(y))\right|.
\]
In this context there are two different objectives that we wish to minimize: the number of outliers and the distortion.
For a compact metric space $\Z = (Z, \rho_Z)$, denote by $\diam(\Z) = \sup_{z, z' \in Z} \rho_Z(z, z')$ the \emph{diameter of $\Z$}.

\begin{definition}[$(\eps,k)$-Outlier embedding]\label{def:bicriteria}
We say that $\X$ admits a \emph{$(\eps,k)$-outlier embedding} into $\Y$ if there exists $\oset \subset X$ with $|\oset|\leq k$ and some $f:X\setminus \oset \to Y$ 
with $\ell_{\infty}$ distortion at most $\eps \diam(\X)$.
We refer to $\oset$ as the \emph{outlier set that witnesses} a $(\eps,k)$-outlier embedding of $X$. 
\end{definition}
Note that the multiplication of the distortion by $\diam(X)$ is to make the parameter $\eps$ scale-free. 
Since an isometry can be trivially achieved by removing all but one points; thus the above notion is well-defined for all $\eps> 0$.
We now state our main results concerning bi-criteria approximation: 

\emph{Bi-criteria outlier embeddings into ultrametrics:}
We obtain a polynomial-time algorithm which given an $n$-point metric space $\X$ and some $\eps>0$ such that $\X$ admits a $(\eps,k)$-outlier embedding into an ultrametric, outputs a $O(\eps \log n, 3k)$-outlier embedding into an ultrametric.

\emph{Bi-criteria outlier embeddings into $\mathbb{R}^d$:}
We present an algorithm which given an $n$-point metric space $\X$ and some $\eps>0$ such that $\X$ admits a $(\eps, k)$-\noembedding{} in $\mathbb{R}^d$, outputs a $(O(\sqrt{\eps}), 2k)$-\noembedding{} of $\X$ into $\mathbb{R}^d$. The algorithm runs in time $O(n^{d+3})$.

\emph{Bi-criteria outlier embeddings into trees:}
Finally we mention that one can easily derive a bi-criteria approximation for outlier embedding into trees by the work of Gromov on $\delta$-hyperbolicity \cite{gromov1987hyperbolic} (see also \cite{chepoi2008diameters}).
Formally, there exists a polynomial-time algorithm which given a metric space $\X$ and some $\eps>0$ such that $\X$ admits a $(\eps,k)$-outlier embedding into a tree, outputs a $(O(\eps \log n), 4k)$-outlier embedding into a tree.
Let us briefly outline the proof of this result:
 $\delta$-hyperbolicity is a four-point condition such that any $\delta$-hyperbolic space admits an embedding into a tree with $\ell_\infty$ distortion $O(\delta \log{n})$, and such an embedding can be computed in polynomial time.
Any metric that admits an embedding into a tree with $\ell_{\infty}$ distortion $\eps$ is $O(\eps)$-hyperbolic.
Thus by removing all 4-tuples of points that violate the $O(\eps)$-hyperbolicity condition and applying the embedding from \cite{gromov1987hyperbolic} we immediately obtain an $(O(\eps\log n), 4k)$-outlier embedding into a tree.
We omit the details.

\subsection{Previous work}
Over the recent years there has been a lot work on approximation algorithms for minimum distortion embeddings into several host spaces and under various notions of distortion.
Perhaps the most well-studied case is that of multiplicative distortion.
For this case, approximation algorithms and inapproximability results have been obtained for embedding into the line \cite{NR2015,badoiu2005approximation,badoiu2005low,fellows2009distortion}, constant-dimensional Euclidean space \cite{BadoiuCIS06,BergOS13,edmonds2010inapproximability,matouvsek2010inapproximability,badoiu2005approximation}, trees \cite{badoiu2007approximation,chepoi2012constant}, ultrametrics \cite{alon2008ordinal}, and other graph-induced metrics \cite{chepoi2012constant}.
We also mention that similar questions have been considered for the case of bijective embeddings \cite{papadimitriou2005complexity,hall2005approximating,kenyon2009low,edmonds2010inapproximability,khot2007hardness}.
Analogous questions have also been investigated for average \cite{dhamdhere2006approximation}, additive \cite{badoiu2003approximation}, $\ell_p$ \cite{ailon2005fitting}, and $\ell_{\infty}$ distortion \cite{farach1995robust,agarwala1998approximability}.

Similar in spirit with the outlier embeddings introduced in this work is the notion of \emph{embeddings with slack} \cite{chan2009metric,chan2006spanners,LammersenSS09}.
In this scenario we are given a parameter $\eps>0$ and we wish to find an embedding that preserves $(1-\eps)$-fraction of all pairwise distances up to a certain distortion.
We remark however that these mappings cannot in general be used to obtain outlier embeddings.
This is because typically in an embedding with slack the pairwise distances that are distorted arbitrarily involve a large fraction of all points.

\subsection{Discussion}
Our work naturally leads to several directions for further research.
Let us briefly discuss the most prominent ones.

An obvious direction is closing the gap between the approximation factors and the inapproximability results for embedding into ultrametrics and trees.
Similarly, it is important to understand whether the running time of the 2-approximation for embedding into Euclidean space can be improved.
More generally, an important direction is understanding the approximability of outlier embeddings into other host spaces, such as planar graphs and other graph-induced metrics.

In the context of bi-criteria outlier embeddings, another direction is to investigate different notions of distortion.
The case of $\ell_\infty$ distortion studied here is a natural starting point since it is very sensitive to outliers.
It seems promising to try to adapt existing approximation algorithms for $\ell_p$, multiplicative, and average distortion to the outlier case.

Finally, it is important to understand whether improved guarantees or matching hardness results for bi-criteria approximations are possible.

\section{Definitions}

A metric space is a pair $\X = (X, \rho)$ where $X$ is a set and $\rho: X \times X \to \mathbb{R}_{\ge 0}$ such that (i) for any $x, y \in X$, $\rho(x, y) = \rho(y, x)\ge 0$, (ii) $\rho(x, y) = 0$ if and only if $x = y$, and (iii) for any $x, y, z \in X$, $\rho(x, z) \le \rho(x, y) + \rho(y, z)$. 
Given two metric spaces $\X = (X, \rho_X)$ and $\Y = (Y, \rho_Y)$, an \emph{embedding of $\X$ into $\Y$} is simply a map $\phi: X \to Y$, and $\phi$ is an \emph{isometric embedding} if for any $x, x' \in X$, $\rho_X(x, x') = \rho_Y( \phi(x), \phi(x'))$. 

In this paper our input is an \emph{$n$-point metric} $(X, \rho)$, meaning that $X$ is a discrete set of cardinality $n$. 
Given an $n$-point metric space $\X = (X, \rho)$ and a value $\alpha \ge 0$, we denote by $\X + \alpha$ the metric space $(X, \rho')$ where for any $x\neq y\in X$ we have $\rho'(x,y)=\rho(x,y)+\alpha$. 


\begin{definition}[Ultrametric space]
A metric space $(X, \rho)$ is an \emph{ultrametric (tree) space} if and only if the following three-point condition holds for any $x, y, z \in X$: 
\begin{align}
\rho(x,y) \leq \max\{\rho(x,z), \rho(z,y)\}. \label{eq:ultrametric}
\end{align}
\end{definition}
\begin{definition}[Tree metric]\label{def:treemetric}
A metric space $(X, \rho)$ is a \emph{tree metric}
if and only if the following four-point condition holds
for any $x, y, z, w \in X$:
\begin{align}\label{eq:four-point}
\rho(x,y)+\rho(z,w) \leq \max\{ \rho(x,z)+\rho(y,w), \rho(x,w)+\rho(y,z)\}.
\end{align}
An equivalent formulation of the four-point condition is that for all $x, y, z, w \in X$, 
the largest two quantities of the following three terms are equal: 
\begin{align}\label{eq:treemetric2}
\rho(x,y)+\rho(z,w), ~~\rho(x,z)+\rho(y,w), ~~\rho(x,w)+\rho(y,z). 
\end{align}
\end{definition} 
In particular, an $n$-point tree metric $(X, \rho)$ can be realized by a weighted tree $T$ such that there is a map $\phi: X \to V(T)$ into the set of nodes $V(T)$ of $T$, and that for any $x, y \in X$, the shortest path distance $d_T(\phi(x), \phi(y))$ in $T$ equals $\rho(x,y)$. In other words, $\phi$ is an isometric embedding of $(X, \rho)$ into the graphic tree metric $(T, d_T)$. 
An ultrametric $(X, \rho)$ is in fact a special case of tree metric, where there is an isometric embedding $\phi: X \to V(T)$ to a rooted tree $(T, d_T)$ such that $\phi(X)$ are leaves of $T$ and all leaves are at equal distance from the root of $T$. 

\section{Approximation algorithms for outlier embeddings}
\label{sec:approx}

In this section we present approximation algorithms for the minimum outlier embedding problem for three types of target metric spaces: ultrametrics, tree metrics, and Euclidean metric spaces. We show in Appendix \ref{appendix:sec:hardness} that finding optimal solutions for each of these problems is NP-hard (and roughly speaking hard to approximate within a factor of $2$ as well). 
In the cases of ultrametric and tree metrics, it is easy to approximate the minimum outlier embedding within constant factor in $O(n^3)$ and $O(n^4)$ time, respectively. The key challenge (especially for embedding into tree metric) is to improve the time complexity of the approximation algorithm to $O(n^2)$, which is optimal. 

\subsection{Approximating outlier embeddings into ultrametrics}


\begin{theorem}
Given an $n$-point metric space $(X, \rho)$, there exists a $3$-approximation algorithm for minimum outlier embedding into ultrametrics, with running time $O(n^2)$.
\end{theorem}

\begin{proof}
We can obtain a polynomial-time $3$-approximation algorithm as follows:
For each triple of points $(x,y,z)\in X^3$, considered in some arbitrary order, check whether it satisfies \eqref{eq:ultrametric}.
If not, then remove $x$, $y$, and $z$ from $X$ and continue with the remaining triples. Let $\oset$ be the set of removed points. 
For every triple of points removed, at least one must be in any optimal solution; therefore the resulting solution $\oset$ is a $3$-approximation.
The running time of this method is $O(n^3)$. We next show how to improve the running time to $O(n^2)$.

Let $X=\{x_1,\ldots,x_n\}$.
We inductively compute a sequence $X_0,\ldots,X_n\subseteq X$, where $X_0$ is set to be $X_0= \emptyset$.
Given $X_{i-1}$ for some $i\in \{1,\ldots,n\}$, assuming the invariance that $X_{i-1}$ is an ultrametric, we compute $X_i$ as follows.
We check whether $(X_{i-1}\cup \{x_i\}, \rho)$ is an ultrametric.
If it is, then we set $X_i = X_{i-1}\cup \{x_i\}$.
Otherwise, there must exist $(x,y,z)\in (X_{i-1}\cup \{x_i\})^3$ that violates \eqref{eq:ultrametric}.
Since $(X_{i-1}, \rho)$ is an ultrametric, it follows that every such triple must contain $x_i$.
Therefore it suffices to show how to quickly find $y,z\in X_{i-1}$ such that $(x_i,y,z)$ violates \eqref{eq:ultrametric}, if they exist.
To this end, let $x_i^*$ be a nearest neighbor of $x_i$ in $X_{i-1}$, that is
$x_i^* = \argmin_{x\in X_{i-1}} \{\rho(x_i, x)\},
$
where we brake ties arbitrarily.
Instead of checking $x_i$ against all possible $y, z$ from $X_{i-1}$, we claim that \eqref{eq:ultrametric} holds for all $(x_i,y,z)$ with $y,z\in X_{i-1}$ if and only if for all $w\in X_{i-1}$ we have
\begin{align}\label{eq:ultra1}
\text{(i)}~~\rho(x_i,x_i^*) &\leq \max\{\rho(x_i,w), \rho(x_i^*,w)\} ~~\text{and}~~\text{(ii)}~~\rho(x_i,w) = \rho(x_i^*,w). 
\end{align}
Indeed, assume that (i) and (ii) above hold for all $w\in X_{i-1}$, yet there exist some $y, z \in X_{i-1}$ such that $(x_i, y, z)$ violates \eqref{eq:ultrametric}, say w.l.o.g., $\rho(x_i, y) > \max \{ \rho(x_i, z), \rho(y, z) \}$. 
Then by (ii) above, we have $\rho(x_i, y) = \rho(x_i^*, y)$ and $\rho(x_i, z) = \rho(x_i^*, z)$, implying that $\rho(x^*_i, y) > \max \{ \rho(x^*_i, z), \rho(y, z) \}$. Hence $(x_i^*, y, z)$ also violates \eqref{eq:ultrametric}, contradicting the fact that $(X_{i-1}, \rho)$ is an ultrametric. 
Hence no such $y, z$ can exist, and \eqref{eq:ultra1} is sufficient to check whether $X_{i-1} \cup \{x_i\}$ induces an ultrametric or not. 

Finally, we can clearly check in time $O(n)$ whether both conditions in \eqref{eq:ultra1} hold for all $w\in X_{i-1}$.
If either (i) or (ii) in \eqref{eq:ultra1} fails then $(x_i,x_i^*,w)$ violates \eqref{eq:ultrametric}, which concludes the proof.
\end{proof}

\subsection{Approximating outlier embeddings into trees}
\label{sec:approx:tree}

We now present a $4$-approximation algorithm for embedding a given $n$-point metric space into a tree metric with a minimum number of outliers. 
Using the four-point condition \eqref{eq:four-point} in Definition \ref{def:treemetric}, 
it is fairly simple to obtain a $4$-approximation algorithm for the problem with running time $O(n^4)$ as follows:
Check all 4-tuples of points $x,y,z,w\in X$.
If the 4-tuple violates the four-point condition, then remove $x,y,z,w$ from $X$.
It is immediate that for any such 4-tuple, at least one of its points much be an outlier in any optimal solution.
It follows that the result is a $4$-approximation.

We next show how to implement this approach in time $O(n^2)$. 
The main technical difficult is in finding a set of violating 4-tuples quickly. 
The high-level description of the algorithm is rather simple, and is as follows.
Let $(X,\rho)$ be the input metric space where $X=\{x_1,\ldots,x_n\}$.
Set $X_1=\{x_1\}$.
For any $i=2,\ldots,n$, we inductively define $X_i\subseteq X$. At the beginning of the $i$-th iteration, we maintain the invariance that $(X_{i-1}, \rho)$ is a tree metric. 
If $(X_{i-1}\cup \{x_i\}, \rho)$ is a tree metric, then we set $X_i=X_{i-1}\cup \{x_i\}$.
Otherwise there must exist $y_i,z_i,w_i\in X_{i-1}$ such that the 4-tuple $(x_i,y_i,z_i,w_i)$ violates the four-point condition; we set $X_i=X_{i-1} \setminus \{y_i,z_i,w_i\}$. 

To implement this idea in $O(n^2)$ time, it suffices to show that for any $i=2,\ldots,n$, given $X_{i-1}$, we can compute $X_i$ in time $O(n)$. 
The algorithm will inductively compute a collection of edge-weighted trees $T_1,\ldots,T_n$, with $T_1$ simply being the graph with $V(T_1) = \{x_1\}$, and maintain the following invariants for each $i\in [1,n]$: 
\begin{description}\denselist
\item[(I-1)] $X_i \subseteq V(T_i)$ and all leaves of $T_i$ are in $X_i$. $(X_i, \rho)$ embeds isometrically into $(T_i, d_{T_i})$; that is, the shortest-path metric of $T_i$ agrees with $\rho$ on $X_i$: for any $x,y\in X_i$, $d_{T_i}(x,y)=\rho(x,y)$.  
\item[(I-2)] At the $i$-th iteration either $X_i = X_{i-1} \cup \{x_i\}$ or $X_i = X_{i-1} \setminus \{y,z,w\}$ where the 4-tuple $\{x_i, y,z,w\}$ violates the four-point condition under metric $\rho$. 
\end{description}

\begin{definition}[Leaf augmentation]\label{def:leafaug}
Given $X'\subseteq X$, let $T$ be a tree with $V(T) \cap X =X'$.
Given $a\in X\setminus X'$ and $u, v\in X'$, 
the \emph{$a$-leaf augmentation of $T$ at $\{u, v\}$} is the tree $T'$ obtained as follows.
Let $P$ be the path in $T$ between $u$ and $v$ (which may contain a single vertex if $u=v$).
Set $r = (\rho(a, u) + \rho(a,v) - \rho(u,v))/2$.
Let $\stemv$ be a vertex in $P$ with $d_P(\stemv,u)=\rho(a, u) - r$; if no such vertex exists then we introduce a new such vertex $\stemv$ by subdividing the appropriate edge in $P$ and update the edge lengths accordingly.
In the resulting tree $T'$ we add the vertex $a$ and the edge $\{a,\stemv \}$ if they do no already exist, and we set the length of $\{a,\stemv\}$ to be $r$.
We call $\stemv$ the \emph{stem} of $a$ (w.r.t.~the leaf augmentation).
When $u=v$ we say that $T'$ is the $a$-leaf augmentation of $T$ at $u$, in which case $T'$ is obtained from $T$ simply by adding $a$ as a leaf attached to $u$, and $u$ is the stem of $a$.
\end{definition}

In what follows, we set $x_i^*$ to be the nearest neighbor of $x_i$ in $X_{i-1}$, that is
\[
x_i^* =\argmin_{x\in X_{i-1}} \rho(x, x_i),
\]
where we break ties arbitrarily. 
Intuitively, if we can build a new tree $T'$ from $T_{i-1}$ so that $(X_{i-1} \cup \{x_i\}, \rho)$ can be isometrically embedded in $T'$, then $T'$ is a $x_i$-leaf augmentation of $T_{i-1}$ at some pair $\{x_i^*, u \}$. 
Our approach will first compute an auxiliary  structure, called $x_i$-orientation on $T_{i-1}$, to help us identify a potential leaf augmentation. We next check for the validity of this leaf augmentation. The key is to produce this candidate leaf augmentation such that if it is not valid, then we will be able to find a 4-tuple violating the four-point condition from it quickly. 
\begin{definition}[$(a,u,v)$-orientation]
Let $X'\subseteq X$ and let $T$ be a tree with $V(T)\cap X=X'$.
Let $a\in X \setminus X'$ and $u,v\in X'$.
The \emph{$(a,u,v)$-orientation} of $T$ is a partially oriented tree $\vec{T}$ obtained as follows:
Let $T'$ be the $a$-leaf augmentation of $T$ at $\{u,v\}$, and let $a'$ be the stem of $a$.
We orient every edge in $P_{uv}\cap E(T')$ towards $a'$, where $P_{uv}$ is the unique path in $T$ between $u$ and $v$. 
All other edges in $E(T) \setminus P_{uv}$ remain unoriented. 

If $E(T)\neq E(T')$ then there exists a unique edge in $e\in P_{uv} \setminus E(T')$ (which is subdivided in $T'$);
this edge $e$ remains undirected in $\vec{T}$. We call this edge the \emph{sink edge w.r.t. $\{u, v\}$}. 
If there is no sink edge, then there is a unique vertex in $P_{uv}$ with no outgoing edges in $P_{uv}$, which we call the \emph{sink vertex w.r.t.~$\{u, v\}$.} 
Note that the sink is the simplex of smallest dimension that contains the stem of $a$ w.r.to the leaf augmentation at $\{u, v\}$. 
\end{definition}

\begin{wrapfigure}{r}{0.2\textwidth}
\vspace*{-0.3in}\fbox{\includegraphics[height=3cm]{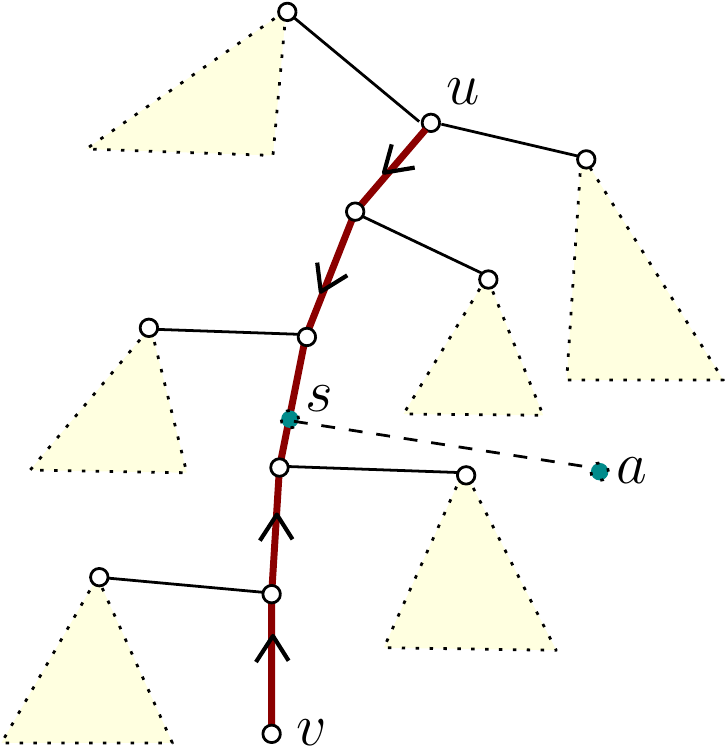}}
\end{wrapfigure}
See the right figure for an example:  where $\stemv$ is stem of $a$ in the leaf augmentation at $\{u, v\}$. The thick path $P_{uv}$ is oriented, other than the sink edge (the one that contains the stem $\stemv$). 

\begin{definition}[$x_i$-orientation]\label{def:globalorientation}
An \emph{$x_i$-orientation of $T_{i-1}$} is any partial orientation $\vec{T}_{i-1}$ of $T_{i-1}$ obtained via the following procedure: 
Consider any ordering of $X_{i-1}$, say $\{v_1, v_2, \ldots, v_\ell\} = X_{i-1}$. 
Start with $\vec{T}_{i-1} = T_{i-1}$,
i.e. all edges in $\vec{T}_{i-1}$ are initialized as undirected and we will iteratively modify their orientation. 
Process vertices $v_1, \ldots, v_{\ell}$ in this order. 
For each $v_j$, denote by $P_j$ the path in $T_{i-1}$ between $v_j$ and $x_i^*$. 
Traverse $P_j$ starting from $v_j$ until we reach either $x_i^*$ or an edge which is already visited. For each unoriented edge we visit, we set its orientation to be the one in the $(x_i, v_j, x_i^*)$-orientation of $T_{i-1}$. 
An edge that is visited in the above process is called \emph{masked}. 
\end{definition} 

Since the above procedure is performed for all leaves of $T_{i-1}$, an $x_i$-orientation will mask all edges. 
However, a masked edge may not be oriented, in which case this edge must be the sink edge w.r.t.~$\{v_j, x_i^*\}$ for some $v_j \in X_{i-1}$. 

\begin{definition}[Sinks]
Given an $x_i$-orientation $\vec{T}$ of tree $T$, a \emph{sink} is either an un-oriented edge, or a vertex $v \in V(T)$ such that all incident edges have an orientation toward $v$. The former is also called a \emph{sink edge w.r.t. $\vec{T}$} and the latter a \emph{sink vertex w.r.t. $\vec{T}$}. 
\end{definition}
It can be shown that each sink edge/vertex must be a sink edge/vertex w.r.t.~$\{v_j, x_i^*\}$ for some $v_j \in X_{i-1}$, and we call $v_j$ a \emph{generating vertex} for this sink. 

\begin{wrapfigure}{r}{0.15\textwidth}
\vspace*{-0.1in}\fbox{\includegraphics[height=3cm]{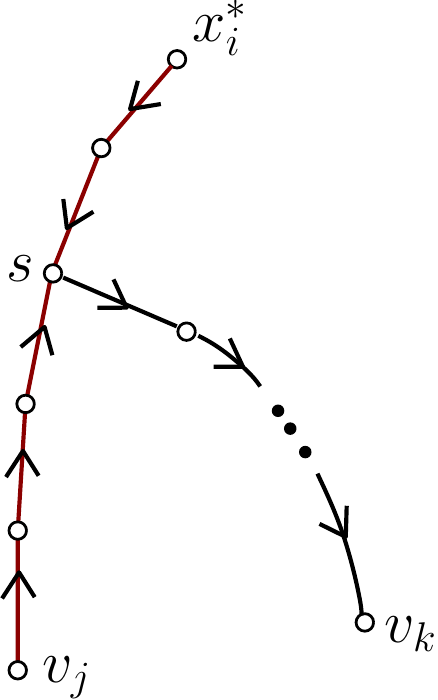}}
\end{wrapfigure}
An $x_i$-orientation may have multiple sinks. We further augment the $x_i$-orientation to record a generating vertex $v_j$ for every sink (there may be multiple choices of $v_j$ for a single sink, and we can take an arbitrary one). 
We also remark that a sink w.r.t. some $\{v_j, x_i^*\}$ may not ultimately be a sink for the global $x_i$-orientation: see the right figure for an example, where $s$ is a sink vertex w.r.t. $\{v_j, x_i^*\}$, but not a sink vertex for the global $x_i$-orientation. 

The proofs of the following two results can be found in Appendix \ref{appendix:sec:approx}. 
\begin{lemma}\label{lem:comp-orientation}
An $x_i$-orientation of $T_{i-1}$ (together with a generating vertex for each sink) can be computed in $O(i)$ time. 
\end{lemma}
\begin{lemma}\label{lem:sinkexists}
Any $x_i$-orientation $\vec{T}_{i-1}$ of $T_{i-1}$ has at least one sink. 
\end{lemma}
%
%
\begin{figure}[htbp]
\begin{center}
\begin{tabular}{ccccc}
\fbox{\includegraphics[height=3.3cm]{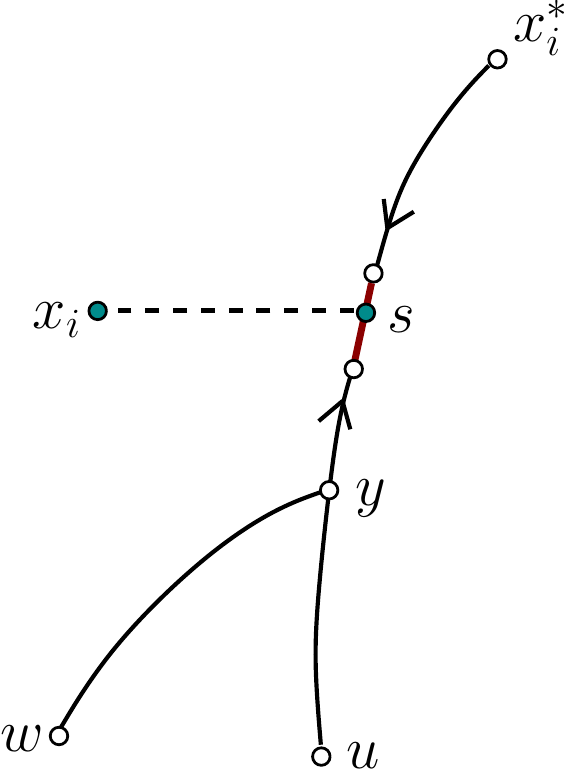}} &\hspace*{0.2in} &\fbox{\includegraphics[height=3.3cm]{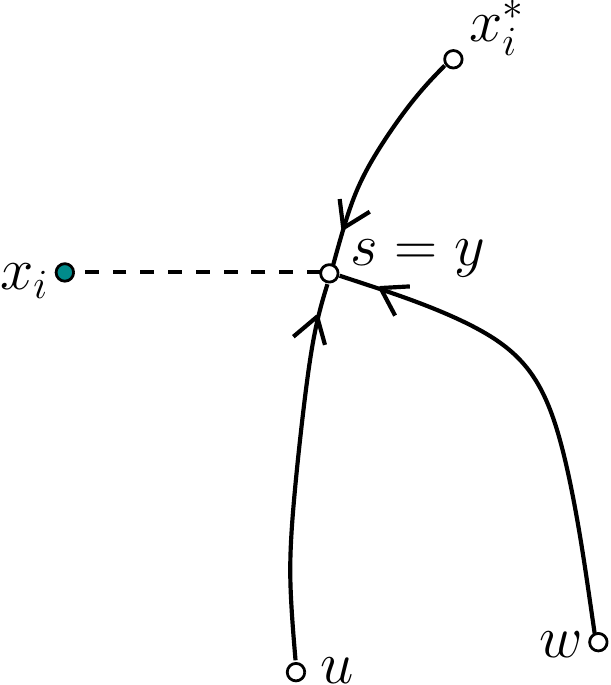}} & \hspace*{0.2in}& \fbox{\includegraphics[height=3.3cm]{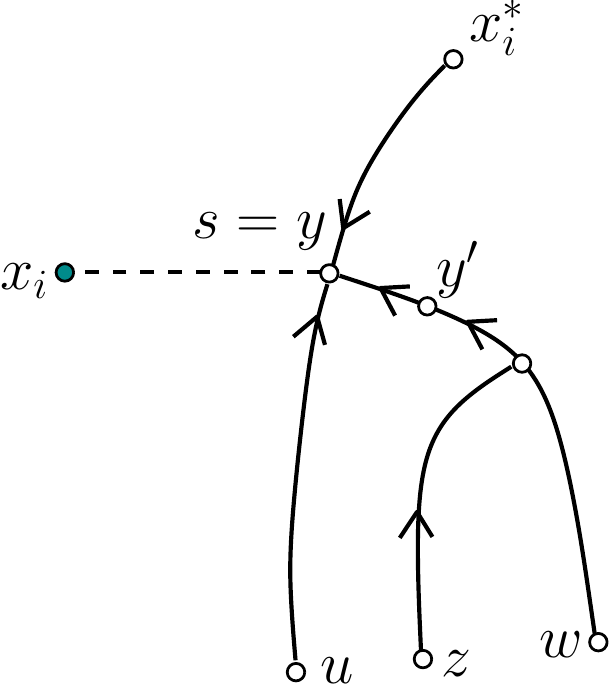}}\\
(a) & & (b) & & (c) 
\end{tabular}
\end{center}
\vspace*{-0.25in}
\caption{{\small (a) An illustration for Case 1 in Lemma \ref{lem:T_i}. Solid edges are from $T_{i-1}$ and the dotted edge connecting $x_i$ to its stem $s$ is in $L$. (b) Case 2 when $y = s$. (c) If $\rho(x_i,w) > d_L(x_i,w)$, then we can find $z$ the processing of which during the construction of $\vec{T}_{i-1}$ assigns the orientation of edge $(y,y')$. 
}
\label{fig:lem:T_i}}
\end{figure}
\begin{lemma}\label{lem:T_i}
For any $i\in \{2,\ldots,n\}$, given $T_{i-1}$, we can compute $T_i$ and $X_i$ satisfying invariants {\sf (I-1)} and {\sf (I-2)} in $O(i$) time.
\end{lemma}
\begin{proof}
It suffices to show that in $O(i)$ time we can either find a $4$-tuple of points in $X_{i-1}\cup \{x_i\}$, that violates the four-point condition, or we can compute a tree having a shortest-path metric that agrees with $\rho$ on $X_{i-1}\cup \{x_i\}$. 
By Lemma \ref{lem:comp-orientation}, we can compute an $x$-orientation $\vec T_{i-1}$ of $T_{i-1}$ in $O(i)$ time. 
Consider any sink of $\vec T_{i-1}$ (whose existence is guaranteed by Lemma \ref{lem:sinkexists}), and let $u$ be its associated generating vertex; 
$u$ must be in $X_{i-1}$. 
Let $L$ be the $x_i$-leaf augmentation of $T_{i-1}$ at $\{x_i^*, u\}$, and let $d_L$ denote the shortest path metric on the tree $L$. 

Since $L$ is the $x_i$-leaf augmentation of $T$, we have for all $w, w' \in X_{i-1}$, $d_L(w, w') = d_{T_{i-1}}(w,w') = \rho(w, w')$ (the last quality is because $(X_{i-1}, \rho)$ embeds isometrically into $T_{i-1}$). 
Thus $d_L$ may only disagree with $\rho$ on pairs of points $x_i, v$, for some $v \in X_{i-1}$. 
We check in $O(i)$ time if, for all $v\in X_{i-1}$, we have $d_L(x_i, v) = \rho(x_i, v)$ via a traveral of $L$ starting from the stem of $x_i$ in $L$. 
If the above holds, then obviously $(X_{i-1}\cup \{x_i\}, \rho)$ embeds isometrically into $L$. We then set $X_i = X_{i-1} \cup \{x_i\}$ and output $T_i = L$. 
Otherwise, let $w \in X_{i-1}$ be such that $d_L(x_i, w) \neq \rho(x_i, w)$. 
We now show that we can find a $4$-tuple including $x_i$ that violates the four-point condition in constant time. 

Let $\stemv$ be the stem of $x_i$ in $L$. Consider $L$ as rooted at $x_i^*$ and let $y \in V(L)$ be the lowest common ancester of $u$ and $w$. Note that $y$ must be a vertex from $V(T_{i-1})$ too. 
Let $P_{vv'}$ denote the unique path in $L$ between any two $v$ and $v'$. 
The vertex $y$ must be in the path $P_{u x_i^*}$. 
\begin{description}\denselist
\item{Case 1: } $y \neq \stemv$.
In this case, $y$ is either in the interior of path $P_{\stemv u}$ or of path $P_{\stemv x_i^*}$. Assume w.o.l.g. that $y$ is in the interior of $P_{\stemv u}$; the handling of the other case is completely symmetric. See Figure \ref{fig:lem:T_i} (a) for an illustration. 
Since $d_L$ is a tree metric, we know that the $4$-tuple $\{x_i, x_i^*, u, w \}$ should satisfy the four-point condition under the metric $d_L$. 
Using the alternative formulation of four-point condition in Definition \ref{def:treemetric}, we have that the largest two quantities of the following three terms should be equal: 
\begin{align}\label{eq:threeterms}
 d_L(x_i, x_i^*) + d_L(u, w), ~~d_L(x_i, u) + d_L(x_i^*, w), ~~d_L(x_i, w) + d_L(x_i^*, u). 
\end{align}
For this specific configuration of $y$, we further have: 
\begin{align}\label{eq:dL}
d_L(x_i, x_i^*) + d_L(u, w) &< d_L(x_i, u) + d_L(x_i^*, w) = d_L(x_i, w) + d_L(x_i^*, u).
\end{align}
On the other hand, by construction, we know that $d_L$ agrees with $\rho$ on $X_{i-1}$. 
Furthermore, since $L$ is the $x_i$-leaf augmentation of $T_{i-1}$ at $\{u, x_i^*\}$, we have that $d_L(x_i, u) = \rho(x_i, u)$ and $d_L(x_i, x_i^*) = \rho(x_i, x_i^*)$. 
Hence \eqref{eq:dL} can be rewritten as 
\begin{align}\label{eq:dL2}
 \rho(x_i, x_i^*) + \rho(u, w) & < \rho(x_i, u) + \rho(x_i^*, w) = d_L(x_i, w) + \rho(x_i^*, u). 
\end{align}
If $\rho(x_i, w) \neq d_L(x_i, w)$, then the largest two quantities of 
\[\rho(x_i, x_i^*) + \rho(u, w), \rho(x_i, u) + \rho(x_i^*, w), \rho(x_i, w) + \rho(x_i^*, u)
\]
can no longer be equal as $\rho(x_i, x_i^*) + \rho(u, w) < \rho(x_i, u) + \rho(x_i^*, w)$. Hence the $4$-tuple $\{x_i, x_i^*, u, w \}$ violates the four-point condition under the metric $\rho$ (by using \eqref{eq:treemetric2}). 
\item{Case 2: $y = \stemv$}, in which case $\stemv$ must be a sink vertex: see Figure \ref{fig:lem:T_i} (b) for an illustration. 
For this configuration of $y$, it is necessary that 
\begin{align}\label{eq:dL3}
 \rho(x_i, x_i^*) + \rho(u, w) & = \rho(x_i, u) + \rho(x_i^*, w) = d_L(x_i, w) + \rho(x_i^*, u). 
\end{align}
Hence if $\rho(x_i, w) > d_L(x_i, w)$, then the $4$-tuple $\{x_i, x_i^*, u, w \}$ violates the four-point condition under the metric $\rho$ because 
\[ \rho(x_i, x_i^*) + \rho(u, w) = \rho(x_i, u) + \rho(x_i^*, w) < \rho(x_i, w) + \rho(x_i^*, u). 
\]
What remains is to find an violating 4-tuple for the case when $\rho(x_i, w) < d_L(x_i, w)$. 

Now imagine performing the $x_i$-leaf augmentation of $T_{i-1}$ at $\{w, x_i^*\}$. We first argue that the stem $\stemv'$ of $x_i$ w.r.t.~$\{w, x_i^*\}$ necessarily lies in $P_{wy}$ in $T_{i-1}$. 
Let $r_w = (\rho(x_i, x_i^*) + \rho(x_i, w) - \rho(x_i^*, w))/2$. 
In the augmented tree $L$, $d_L(x_i, y) = (\rho(x_i, x_i^*)+\rho(x_i, u) - \rho(x_i^*, u))/2$. 
Combing \eqref{eq:dL3} and $\rho(x_i, w) < d_L(x_i, w)$ we have that $r_w < d_L(x_i, y)$. 
On the other  hand, following Definition \ref{def:leafaug}, the position of $\stemv'$ is such that $d_{T_{i-1}}(x_i^*, \stemv') = \rho(x_i^*, x_i) - r_w$, while the position of $y = \stemv$ was that $d_{T_{i-1}} (x_i^*, s) = \rho(x_i^*, x_i) - d_L(x_i,  y)$. 
It then follows that $\stemv'$ must lie in the interior of path $P_{wy}$. 

Since the stem of $x_i$ w.r.t.~$\{w, x_i^*\}$ is in $P_{wy}$, it means that before we process $w$ in the construction of the $x_i$-orientation $\vec T_{i-1}$, there must exist some other leaf $z \in V(T_{i-1})$ such that the process of $z$ assigns the orientation of the edge $(y, y') \subset P_{wy}$ to be towards $y$; See Figure \ref{fig:lem:T_i} (c).  This is because if no such $z$ exists, then while processing $w$, we would have oriented the edge $(y,  y')$ towards stem $\stemv'$, thus towards $y'$, as the stem $\stemv'$ is in $P_{wy'}$. 
The point $z$ can be identified in constant time if during the construction of $\vec T_{i-1}$, we also remember, for each edge, the vertex the processing of which leads to orienting this edge. Such information can be easily computed in $O(i)$ time during the construction of $\vec T_{i-1}$. 

Now consider $z$. If $d_L(x_i, z) = \rho(x_i, z)$, then one can show that $y = \stemv$ is necessarily the stem for $x_i$ w.r.t.~$\{x_i, z\}$ as well (by simply computing the position of the stem using Definition \ref{def:leafaug}). In this case, considering the 4-tuple $\{x_i, x_i^*, z, w\}$, we are back to Case 1 (but for this new 4-tuple), which in turn means that this 4-tuple violates the four-point condition. Hence we are done. 

If $d_L(x_i, z) \neq \rho(x_i, z)$, then since we orient the edge $(y, y')$ towards $y$ during the process of leaf $z$, the stem of $x_i$ of the leaf augmentation at $\{x_i, z\}$ is in the path $P_{yx_i^*}$. 
By an argument similar to the proof that $\stemv'$ is in the interior of $P_{wy}$ above, we can show that $\rho(x_i, z) > d_L(x_i, z)$. 
Now consider the 4-tuple $\{x_i, x_i^*, z, u\}$: 
this leads us to an analogous case when $\rho(x_i, w) > d_L(x_i, w)$ for the 4-tuple $\{x_i, x_i^*, u, w \}$. 
Hence by a similar argument as at the beginning of Case 2, we can show that $\{x_i, x_i^*, z, u\}$ violates the four-point condition under metric $\rho$. 
\end{description}

Putting everything together, in either case, we can identify a 4-tuple $U$, which could be $\{x_i, x_i^*, u, w \}$, $\{x_i, x_i^*, z, w\}$, or $\{x_i, x_i^*, z, u\}$ as shown above, that violates the four-point condition under metric $\rho$. 
We simply remove these four points, adjust the resulting tree to obtain $T_i$ and set $X_i = X_{i-1} \setminus U$. 
The overall algorithm takes $O(i)$ time as claimed. This proves the lemma. 
\end{proof}

\begin{theorem}\label{thm:outliertree}
There exists a $4$-approximation algorithm for minimum outlier embedding into trees, with running time $O(n^2)$.
\end{theorem}
\begin{proof}
By Lemma \ref{lem:T_i} and induction on $i=1,\ldots,n$, it follows immediately that we can compute $T_n$ in time $O(n^2)$. 
By invariant {\sf (I-1)}, the output $(X_n, \rho)$ is a tree metric as it can be isometrically embedded into $T_n$. 
Furthermore, by invariant {\sf (I-2)}, each 4-tuple of points we removed forms a violation of the four-point condition, and thus must contain at least one point from any optimal outlier set. 
As such, the total number of points we removed can be at most four times the size of the optimal solution. 
Hence our algorithm is a $4$-approximation as claimed. 
\end{proof}

\subsection{Approximating outlier embeddings into $\mathbb{R}^d$}
\label{subsec:Rdmetricapprox}

In this section, we present a $2$-approximation algorithm for the minimum outlier embedding problem into the Euclidean space $\mathbb{R}^d$ in polynomial time, which matches our hardness result in Appendix \ref{appendix:sec:hardness}. 
Given two points $p, q \in \mathbb{R}^d$, let $d_E(p, q) = \| p - q \|_2$ denote the Euclidean distance between $p$ and $q$. 
\begin{definition}[$d$-embedding]
Given a discrete metric space $\X = (X, \rho)$, an \emph{$d$-embedding} of $\X$ is simply an isometric embedding $\phi: X \to \mathbb{R}^d$ of $\X$ into $(\mathbb{R}^d, d_E)$; that is, for any $x, y \in X$, $\rho(x,y) = d_E(\phi(x), \phi(y))$. 
We say that $\X$ is \emph{strongly $d$-embeddable} if it has a $d$-embedding, but cannot be isometrically embedded in $\mathbb{R}^{d-1}$. In this case, $d$ is called the \emph{embedding dimension of $\X$}. 
\end{definition} 

The following is a classic result in distance geometry of Euclidean spaces, see e.g \cite{Blumenthal70,SS86}. 
\begin{theorem}\label{thm:EuclideanEmbed}
The metric space $\X = (X, \rho)$ is strongly $d$-embeddable in $\mathbb{R}^d$ if and only if there exist $d+1$ points, say $X_d = \{ x_0, \ldots, x_d \}$, such that:\\
$~~~~${(i)}~ $(X_d, \rho)$ is strongly $d$-embeddable; and  \\
$~~~~$(ii)~for any $x, x' \in X \setminus X_d$, $(X_d \cup \{x, x'\}, \rho)$ is $d$-embeddable. 
\end{theorem}

Furthermore, given an $m$-point metric space $(X, \rho)$, it is known that one can decide whether $(X, \rho)$ is embeddable in some Euclidean space by checking whether a certain $m \times m$ matrix derived from the distance matrix $\rho$ is positive semi-definite, and the rank of this matrix gives the embedding dimension of $\X$; see e.g.~\cite{PR03}. 

Following Theorem \ref{thm:EuclideanEmbed}, one can easily come up with a $(d+3)$-approximation algorithm for minimum outlier embedding into $\mathbb{R}^d$, by simply checking whether each $(d+3)$-tuple of points is $d$-embeddable, and if not, removing all these $d+3$ points. Our main result below is an $2$-approximation algorithm within the same running time. 
In particular, Algorithm \ref{alg:outlierRd} satisfies the requirements of Theorem \ref{thm:outlierRd}, and the proof is in Appendix \ref{appendix:sec:approx}. 
\begin{theorem}\label{thm:outlierRd}
Given an $n$-point metric space $(X, \rho)$, for any $d\geq 1$, there exists $2$-approximation algorithm for minimum outlier embedding into $\mathbb{R}^d$, with running time $O(n^{d+3})$.
\end{theorem}

\begin{algorithm}[htbp]
\caption{$2$-approximation outlier embedding in $\mathbb{R}^d$. \label{alg:outlierRd}}
\begin{algorithmic}
\Require An $n$-point metric space $(X, \rho)$ 
\Ensure A set of outliers $\widehat{\oset}$  \\
~~~~~~~~~~~~Initialize the set of candidate outlier sets $\C$ to be empty. 
\State {\sf (Step-0)}~~For each $d+1$ number of distinct points $Y_d = \{y_0, \ldots, y_d \} \subset X$, perform the following: \\
~~~~~~~~~~~~~Initialize sets $Z$ and $\oset$ to be the empty set. 
\begin{description}\denselist
\item ~~~~{\sf (Step-1)}~~Check whether $(Y_d, \rho)$ is $d$-embeddable in $\mathbb{R}^d$: If not, return to (Step-0). Otherwise, compute its embedding dimension $d'$; note, $d' \le d$. 
\item ~~~~{\sf (Step-2)}~~For each remaining point $x \in X \setminus Y_d$, check whether $(Y_d \cup \{s \}, \rho)$ is $d'$-embeddable: If yes, insert $x$ to $Z$; otherwise, insert $x$ to the outlier set $\oset$.
\item ~~~{\sf (Step-3)}~~Construct a graph $G = (Z, E)$, where $(z, z') \in E$ if $(Y_d \cup \{z, z'\}, \rho)$ is {\bf not} $d'$-embeddable. Compute a $2$-approximation $Z' \subset Z$ of the vertex cover of $G$. Set $\oset = \oset \cup Z'$, and we add the set $\oset$ to the collection of candiate outlier sets $\C$.
\end{description}
{\sf (Step-4)} Let $\widehat \oset$ be the set from $\C$ with smallest cardinality. We return $\widehat \oset$ as the outlier set. 
\end{algorithmic}
\end{algorithm}

\paragraph{Hardness results.} 
In Appendix \ref{appendix:sec:hardness}, we show that  the minimum outlier embedding problems into ultrametrics, trees and Euclidean space are all NP-hard, by reducing the Vertex Cover problem to them in each case. In fact, assuming the unique game conjecture, it is NP-hard to approximate each of them within $2-\nu$, for any positive $\nu$. 
For the case of minimum outlier embedding into Euclidean space, we note that our $2$-approximation algorithm above matches the hardness result.

\section{Bi-criteria approximation algorithms}
\label{sec:bicriteria}
\newcommand{\origin}		{{\mathbf o}}
\newcommand{\myconst}		{\mathrm{C}}

\subsection{Bi-criteria approximation for embedding into ultrametrics}
Let $T=(X,E)$ be a tree with non-negative edge weights.
The \emph{ultrametric induced by T} is the ultrametric $\U=(X,\delta)$ where for every $x,y\in X$ we have that $\delta(x,y)$ is equal to the maximum weight of the edges in the unique $x$-$y$ path in $T$
(it is easy to verify that the metric constructed as such is indeed an ultrametric \cite{farach1995robust}). 
Given an metric space $\X = (X, \rho)$, we can view it as a weighted graph and talk about its minimum spanning tree (MST). 
The following result is from \cite{farach1995robust}.

\begin{lemma}[Farach, Kannan and Warnow \cite{farach1995robust}]\label{lem:ultra_opt}
Let $\X=(X,\rho)$ be a metric space and let $\U^*$ be an ultrametric minimizing $\|\X-\U^*\|_{\infty}$.
Let $\widehat{\U}$ be ultrametric induced by a MST of $\X$.
Then there exists $\alpha\geq 0$, such that $\U^* = \widehat{\U}+\alpha$.

In particular,  let $\beta=\|\widehat{\U}-\X\|_{\infty}$.
Then $\U^* = \widehat{\U}+\beta/2$. This further implies that $\|\U^*-\X\|_{\infty} = \|\widehat{\U}-\X\|_{\infty}/2 =\beta/2$.
\end{lemma}


\begin{theorem}
There exists a polynomial-time algorithm which given an $n$-point metric space $\X=(X,\rho)$, $\eps\geq 0$, and $k\geq 0$, such that $\X$ admits a $(\eps, k)$-outlier embedding into an ultrametric, outputs a $(O(\eps \log(n)), 3 k)$-outlier embedding into an ultrametric.
\end{theorem}

\begin{proof}
For simplicity, assume that the diameter $\diam(\X) = 1$. 
The algorithm is as follows.
We first enumerate all triples $\{x,y,z\}\in {X \choose 3}$.
For any such triple, if $\rho(x,y) > \max\{\rho(x,z), \rho(z,y)\} + 2\eps$, then we remove $x$, $y$, and $z$ from $X$.
Let $X'$ be the resulting point set.
We output the ultrametric $\U'$ induced by an MST of the metric space $(X',\rho)$.
This completes the description of the algorithm. 

It suffices to prove that the output is indeed an $(O(\eps \log(n)), 3 k)$-outlier embedding.
Let $\oset^*\subseteq X$, with $|\oset^*| \leq k$, be such that $(X\setminus \oset^*, \rho)$ admits a $(\eps, 0)$-outlier embedding into an ultrametric.
Let $\{x,y,z\}\in {X \choose 3}$ be such that $\rho(x,y) > \max\{\rho(x,z), \rho(z,y)\} + 2\eps$.
It follows by Lemma \ref{lem:ultra_opt} that $(\{x,y,z\}, \rho)$ does not admit a $(\eps, 0)$-outlier embedding into an ultrametric.
Thus, $\oset^*\cap \{x,y,z\}\neq \emptyset$.
It follows that $|X\setminus X'| \leq 3\cdot |\oset^*| \leq 3k$.
In other words, the algorithm removes at most $3k$ points.

It remains to bound the distortion between $\U'$ and $(X', \rho)$.
Let $T$ be the MST of $(X',\rho)$ such that the algorithm outputs the ultrametric $\U'=(X', \delta)$ induced by $T$.
We will prove by induction on $i\geq 0$, that for all $x,y\in X'$, if the $x$-$y$ path in $T$ contains at most $2^i$ edges, then 
$\rho(x,y) \leq \delta(x,y) + 2i \cdot \eps.$

For the base case $i=0$ we have that $\{x,y\}\in E(T)$. Since $T$ is the minimum spanning tree of $(X', \rho)$, it follows that  $\delta(x,y)=\rho(x,y)$, proving the base case. 
For the inductive step, let $x,y\in X'$ such that the $x$-$y$ contains at most $2^i$ edges, for some $i \geq 1$.
Let $w\in X'$ be such that $w$ is in the $x$-$y$ path in $T$, and moreover the $x$-$w$ and $w$-$z$ paths in $T$ have at most $2^{i-1}$ edges each.
Since $\{x,y,w\} \in X'$, it follows that the triple $\{x,y,w\}$ was not removed by the algorithm, and thus
\begin{align}
\rho(x,y) &\leq \max\{\rho(x,w), \rho(w,y)\} + 2\eps. \label{eq:ultra0}
\end{align}
By the inductive hypothesis we have
\begin{align}
\rho(x,w) &\leq \delta(x,w) + 2(i-1)\cdot \eps,~~\text{and}~~\rho(w,y) \leq \delta(w,y) +  2(i-1)\cdot \eps. \label{eq:ultra1}
\end{align}
By \eqref{eq:ultra0} and \eqref{eq:ultra1} we get
\begin{align*}
\rho(x,y) &\leq \max\{\delta(x,w), \delta(w,y)\} + 2(i-1) \cdot \eps + 2\eps 
\leq \delta(x,y) +  2 i \cdot \eps. 
\end{align*}
Hence $\U'$ is an $(2\eps \log n, 0)$-embedding of $(X', \rho)$. This, together with the bound on $|X \setminus X'|$ concludes the proof. 
\end{proof}

\subsection{Bi-criteria approximation for embedding into $\mathbb{R}^d$}

\begin{theorem}\label{thm:bicriteriaRdFinal}
Given an $n$-point metric space $\X = (X, \rho)$, if $\X$ admits a $(\eps, k^*)$-\noembedding{} in $\mathbb{R}^d$, 
then we have an algorithm that outputs an $(O(\sqrt{\eps}), 2k^*)$-\noembedding{} of $\X$ in $\mathbb{R}^d$ in $O(n^{d+3})$ time. 
Here the big-$O$ notation hides constants depending on the dimension $d$. 
\end{theorem}
We only state the main result here, and the details can be found in Appendix \ref{appendix:sec:bicriteriaRd}. 
In particular, the high level structure of our algorithm parallels that of Algorithm \ref{alg:outlierRd}, which intuitively is an algorithm for the special case when $\eps = 0$. 
However, the technical details here are much more involved so as to tackle several issues caused by the near-isometric embedding.

\bibliographystyle{plain}
\bibliography{main}

\newpage
\appendix

\section{Missing details from Section \ref{sec:approx}}
\label{appendix:sec:approx}

\paragraph{Proof of Lemma \ref{lem:comp-orientation}.}
Start with an arbitrary ordering of $X_{i-1}$, and follow the procedure laid-out in Definition \ref{def:globalorientation}. 
Specifically, for each $v_j$, we can compute the distance $r_j = (\rho(v_j, x_i) + \rho(x_i^*, v_j) - \rho(x_i, x_i^*))/ 2$ in constant time. Hence we know the distance $l_j = \rho(v_j, x_i)  - r_j$ from $v_j$ to the stem of $x_i$ w.r.t. $\{v_j, x_i^*\}$. 
Now, starting from $v_j$, we follow the unique path towards the root $x_i^*$, till we meet an masked edge or we  meet the root $x_i^*$. (In other words, we do not traverse the entire path $P_j$ from $v_j$ to $x_i^*$.) 
For any edge we traverse, first we orient it towards $x_i^*$ till we meet an edge $(w, w')$ such that $\rho(w, v_j) < l_j$ and $\rho(w', v_j) > l_j$, or a vertex $w$ such that $\rho(w,v_j) = l_j$. In this case, we set this edge or vertex to be a sink (as it would contain the stem of $x_i$ at $\{v_j, x_i^*\}$) and associate $v_j$ as its generating vertex. 
After the sink, we orient each edge to be towards the sink, till we meet an masked edge. 
The entire procedure to process all $v_j$'s takes $O(|V(T_{i-1})|)$ time -- since all leaves of $T_{i-1}$  are from $X_{i-1}$, we have that $O(|V(T_{i-1})|) = O(|X_{i-1}|) = O(i)$. 

Note that in this procedure, some sink vertex w.r.t. certain $\{v_j, x_i^*\}$ may ultimately not be a sink vertex for the final $x_i$-orientation; recall the figure above this lemma. So we may associate generating vertex information for some non-sink vertices. 

\paragraph{Proof of Lemma \ref{lem:sinkexists}.}
Consider the ordering of leaf vertices of $T_{i-1}$, denoted by $\{ v_1, \ldots, v_\ell \}$, that gives rise to the $x_i$-orientation $\vec T_{i-1}$. 
Following the procedure laid out in the proof of Lemma \ref{lem:comp-orientation}, 
after processing $v_1$, we obtained the $(x_i, v_1, x_i^*)$-orientation along the path $P_{v_1x_i^*}$ (from $v_1$ to $x_i^*$ in $T_{i-1}$). Hence at this point, there must be a sink w.r.t. $\{v_1, x_i^*\}$. If this is a sink edge, since this edge is masked, it remain unoriented throughout the entire procedure; that is, this edge will still be a sink edge for the final $x_i$-orientation $\vec T_{i-1}$. 
Hence $\vec T$ has at least one sink, and we are done. 

Otherwise, consider the sink vertex $w_1$ w.r.t. $\{v_1, x_i^*\}$: Since $w_1$ is not a sink vertex for $\vec T_{i-1}$, there must exist an outgoing edge $(w_1, u)$ oriented towards $u$. 
In other words, the orientation is away from the root $x_i^*$. 
Follow any outgoing edge from $u$ and continue in this manner: we have to either meet a sink edge / vertex (at which point we cannot find another outgoing edge to continue), or we reach a leaf $v_j$ of $T_{i-1}$. All edges we visited during this process are oriented away from the root $x_i^*$. 
Since when we start with each $v_j$, we always first orient edges towards the root $x_i^*$, this means such a path cannot lead us to a leaf $v_j$, and thus we have to stop at a sink of $\vec T_{i-1}$. 
Putting everything together, $\vec T_{i-1}$ has at least one sink. 

\paragraph{Proof of Theorem \ref{thm:outlierRd}.}
We claim that the output outlier set $\widehat{\oset}$ of Algorithm \ref{alg:outlierRd} stated in the main text $2$-approximates the minimum outlier $d$-embedding of $(X, \rho)$. 
Specifically, let $\oset^*$ be an optimal outlier set. We will show: 
(a) $(X \setminus \widehat{\oset}, \rho)$ is $d$-embeddable; and 
(b) $|\oset^*| \le |\widehat \oset| \le 2|\oset^*|$. 

To prove (a), let $\widehat Y_d$ be the $d+1$ points in {\sf (Step-0)} that gives rise to $\widehat \oset$, and $\widehat d (\le d)$ the embedding dimension of $(\widehat Y_d, \rho)$. Set $Z_d  = X \setminus (\oset \cup \widehat Y_d)$. 
{\sf (Step-3)} ensures that for any two points $z, z' \in Z_d$, $(\widehat Y_d \cup \{z, z' \}, \rho)$ must be $\widehat d$-embeddable. This is because any pair $z, z'$ that violates this condition will form an edge in the graph $G$ constructed in Step-3, and subsequently, the vertex cover will contain at least one point from such a pair. 
Statement (a) then follows from Theorem \ref{thm:EuclideanEmbed}. 

We now prove statement (b). The left inequality is evident. Let $Z^* = X \setminus \oset^*$. Since $(Z^*, \rho)$ is $d$-embeddable, by Theorem \ref{thm:EuclideanEmbed}, there exists $d^*+1$ points, say $Y^* = \{ y_0, \dots, y_{d^*} \}$ with $d^* \le d$ being the embedding dimension of $Z^*$, such that $(Y^*, \rho)$ is strongly $d^*$-embeddable. 

Now consider the time when the algorithm enumerates some $Y_d$ such that $Y_d \subseteq Z^*$ and $Y^*\subseteq Y_d$ in {\sf (Step-0)}; $Y_d = Y^*$ if $d^* = d$. 
In the subsequent {\sf (Step-2)}, any point inserted into $\oset$ violates condition (ii) of Theorem \ref{thm:EuclideanEmbed} and thus must belong to the optimal outlier set $\oset^*$ as well. 
In {\sf (Step-3)}, again by Theorem \ref{thm:EuclideanEmbed} (ii), for any edge $(z, z') \in E$, at least one of its endpoint is necessarily in $\oset^*$. Hence a smallest set of outliers consistent with $Y_d$ would consist of the points removed in {\sf (Step-2)} and the vertex cover of the graph $G$ constructed in {\sf (Step-3)}. 
Since the algorithm computes a $2$-approximation of the vertex cover, it follows that the set $\oset$ at the end of {\sf (Step-3)} contains at most $2 |\oset^*|$ points. 
Since $|\widehat \oset| \le |\oset|$, we then have that $|\widehat \oset| \le 2|\oset^*|$. 

Hence the output outlier set $\widehat \oset$ is a $2$-approximation for the minimum outlier embedding of $(X, \rho)$ to $\mathbb{R}^d$ as claimed. 

Finally, we analyze the time complexity of our algorithm. Steps 1 -- 3 will be executed $O(n^{d+1})$ number of times. 
For each $Y_d$, {\sf (Step-1)} takes time polynomial in $d$ due to the discussion below Theorem \ref{thm:EuclideanEmbed}. {\sf (Step-2)} performs $O(n)$ number of embedding test, each for $O(d)$ points. Hence it takes $O(n d^{O(1)})$ total time. In {\sf (Step-3)}, it takes $O(n^2 d^{O(1)})$ time to construct the graph $G = (Z, E)$. A $2$-approximation of the vertex cover can easily be computed in $O(|V|+|E|)$ time \cite{CLRS}. 
Putting everything together, the running time is $O(n^{d+3} d^{O(1)}) = O(n^{d+3})$, where the big-O notation hides terms polynomial in $d$.

\section{Hardness results}
\label{appendix:sec:hardness}

In this section, we show that the minimum outlier embedding problems into ultrametrics, trees, and Euclidean space are all NP-hard, by reducing the vertex cover problem to them in each case. In fact, it is NP-hard to approximate each of them within $2-\eps$, for any $\eps>0$, unless the unique game conjecture is true. 
For the case of minimum outlier embedding into Euclidean space, we note that our $2$-approximation algorithm from the previous section almost matches the hardness result here. 

All the hardness results are obtained via reduction from vertex cover. The high-level approach is very similar, although the specific reduction is different in each case. Below we present the result for tree metric first. 

\begin{theorem}\label{thm:treemetrichard}
The problem of minimum outlier embedding into a tree metric is NP-hard. 
Furthermore, assuming the Unique Games Conjecture, it is NP-hard to approximate within a factor of $2-\eps$, for any $\eps>0$. 
\end{theorem}
\begin{proof}
We use a reduction from Vertex Cover.
This problem is known to be NP-hard.
Furthermore, assuming the Unique Games Conjecture, it is also hard to approximate within a multiplicative factor of $2-\eps$ for any $\eps>0$ \cite{KR08}.

Let $G=(V,E)$ be an instance to Vertex Cover.
The goal is to find the smallest subset of vertices of $V$ covering all edges in $E$. 
We construct a discrete metric space $(X, d_X)$ as follows: Let $\nu$ be an arbitrarily small positive number. 
\begin{itemize}\denselist
\item For any $v_i \in V$, add two nodes $x_i, y_i \in X$. Set $X = \{ x_i, y_i \}_{i\in [1, n]} \cup \{ o\}$. 
\item Set $d_X(x_i, y_i) = 1$, $d_X(o, x_i) = 2$, and $d_X(o, y_i) =1$ for each $i\in [1, n]$. Set $d_X(y_i, y_j) = 2$ and $d_X(x_i, y_j) = 3$ for $i \neq j$. 
\item Set $d_X(x_i, x_j) = 4 - \nu$ if $(v_i, v_j) \in E$ is an edge in $E$; otherwise, set $d_X(x_i, x_j) = 4$ if $(v_i, v_j) \notin E$. 
\end{itemize}
It is easy to verify that the above description indeed specifies a metric space.
Intuitively, $(X, d_X)$ is close to the tree metric induced by a star-shaped tree: $o$ is the root, tree edges are $(o, x_i)$ and $(x_i, y_i)$, for $i\in [1, n]$, and each edge has weight $1$. 
Specifically, we claim that the optimal solution to minimum outlier embedding of $(X, d_X)$ into a tree-metric has the same size as the optimal vertex cover for $G=(V, E)$. 

Indeed, let $\oset^* \subset X$ be the smallest set of outliers that one needs to remove so that the metric space restricted to $X \setminus \oset^*$ is a tree metric. 
Let $C^* \subset V$ denote the optimal vertex cover for $G = (V, E)$. 
It is easy to see that $|\oset^*| \le |C^*|$, as removing the nodes $x_i$ corresponding to vertices in $C^*$ will result in a tree metric induced by the star rooted at $o$ described above. 
We next show that $|C^*| \le |\oset^*|$. 

Given $\oset^*$, construct a subset $\hat{V} \subset V$ as $\hat{V} = \{ v_i \in V \mid x_i \in \oset^* $ or $y_i \in \oset^* \}$. By construction, $|\hat V| \le |\oset^*|$. 
We now show that $\hat{V}$ forms a vertex cover for $G=(V, E)$, which in turn implies $|C^*| \le |\oset^*|$. 
Indeed, suppose there exists an edge $(v_i, v_j) \in E$ such that $v_i \notin \hat{V}$ and $v_j \notin \hat{V}$. 
This means that $x_i, y_i, x_j, y_j \in X \setminus \oset^*$. 
However, it is easy to check that these four points violate the four-point condition as in \eqref{eq:four-point}, as $d_X(x_i, y_j) + d_X(x_j, y_i) = 6$, which is strictly larger than $d_X(x_i, x_j) + d_X(y_j, y_i) = 6 - \nu$, as well as larger than $d_X(x_i, y_i) + d_X(y_j, x_j) = 4$. Contradiction. Hence either $v_i$ or $v_j$ must be included in $\hat{V}$. Hence $|C^*| \le |\hat{V}| \le |\oset^*|$. 

Combining $|\oset^*|\le |C^*|$ and $|C^*| \le |\oset^*|$, we have $|C^*| = |\oset^*|$. The claimed hardness results, both for computing and for approximating the tree-metric outlier-embedding problem thus follow from the hardness of Vertex Cover. 
\end{proof}

\begin{theorem}\label{thm:ultrametrichard}
The problem of minimum outlier embedding into an ultrametric is NP-hard. 
Furthermore, assuming the Unique Games Conjecture, it is NP-hard to approximate within a factor of $2-\eps$, for all $\eps>0$.
\end{theorem}
\begin{proof}
As for the case of embedding into trees, we give a reduction from Vertex Cover. 
Specifically, let $G = (V, E)$ be an instance of Vertex Cover.
We construct a discrete metric space $(U, d_U)$ as follows: Let $\nu$ be a sufficiently small positive number. 
\begin{itemize}\denselist
\item For each $v_i \in V$, $i\in [1, n]$, add a pair of nodes $u_i, w_i \in U$. Hence $U = \{u_i, w_i \}_{i\in [1,n]}$. 
\item Set $d_U(u_i, w_i) = 2\nu$, and $d_U(u_i, w_j) = 1$, $d_U(w_i, w_j) =1$ for any $i\neq j$. For $i\neq j$, set $d_U(u_i, u_j) = 1-\nu$ if $(v_i, v_j) \in E$; otherwise, set $d_U(u_i, u_j) = 1$. 
\end{itemize}
It is easy to verify that the above description indeed specifies a metric. 
Let $\oset^*$ denote the optimal set of outliers so that $d_U$ restricted to $U \setminus \oset^*$ gives rise to an ultrametric. 
Let $C^*$ denote the optimal vertex cover. 
Similar to the proof of Theorem \ref{thm:treemetrichard}, we will show that $|C^*| = |\oset^*|$. 
First, it is easy to see that by removing those vertices $u_i$ such that $v_i \in C^*$, the resulting metric is an ultrametric. Hence $|\oset^*| \le |C^*|$. 
It remains to show that $|C^*| \le |\oset^*|$. 

To this end, construct $\hat{V} = \{ v_i \in V \mid u_i \in \oset^* $ or $w_i \in \oset^* \}$; obviously, $|\hat{V}|\le |\oset^*|$. 
We claim that $\hat{V}$ forms a vertex cover for $G=(V,E)$. Specifically, assume this is not the case and there exists an edge $(v_i, v_j)\in E$ such that $v_i \notin \hat{V}$ and $v_j \notin \hat{V}$. 
Then $u_i, w_i, u_j \in U \setminus \oset^*$. 
However, these three points violate the three-point condition for ultrametrics as in \eqref{eq:ultrametric}, as $d_U(w_i, u_j) = 1$, which is strictly larger than $\max\{ d_U(w_i, u_i) = 2\nu, d_U(u_j, u_i) = 1-\nu \}$, which is a contradiction. Hence $\hat{V}$ is a vertex cover and thus $|C^*| \le|\hat{V}|\le |\oset^*|$. 
Overall, $|C^*| = |\oset^*|$ and the theorem then follows. 
\end{proof}

We next present a hardness result for outlier embedding into constant-dimensional Euclidean space. 

\begin{theorem}
\label{thm:linemetrichard}
For any $d\geq 2$, 
the problem of minimum outlier embedding into $d$-dimensional Euclidean space is NP-hard. 
Furthermore, assuming the Unique Games Conjecture, it is NP-hard to approximate within a factor of $2-\eps$, for any $\eps>0$.
\end{theorem}

\begin{proof}
We first present the construction for $d=2$, and then discuss how to extend it to the case $d>2$.
As before, we give a reduction from Vertex Cover. 
Specifically, given a graph $G = (V, E)$, we construct a discrete metric space $(X, d_X)$ as follows: Let $\nu$ be a sufficiently small positive real number. 
Set $X = \{x_1, \ldots, x_n, y_1, \ldots, y_n, z_1, \ldots, z_n \}$. 
\begin{itemize}\denselist
\item 
We define:
\begin{enumerate}\denselist
\item for any $i\in [1, n]$, $x_i = (i, -1), y_i = (i, 1), $and $z_i = (i, 0)$.
\end{enumerate}
\item We assign distances $d_X$ as follows: 
\begin{enumerate}\denselist
\item $d_X(x_i, z_i) = 1$, $d_X(y_i, z_i) = 1$ and $d_X(x_i, y_i) =2$, for any $i \in [1, n]$.
\item For any $i, j \in [1, n]$, for $a, b \in \{x, y\}$, $d_X(a_i, b_j) = \|a_i - b_j \|$. 
\item For any $i, j \in [1,n]$, $d_X(z_i, x_j) = \| z_i - x_j \|$, $d_X(z_i, y_j) = \|z_i - y_j\|$. 
\item For $i \neq j$, if $(v_i, v_j) \in E$, then $d_X(z_i, z_j) = |j - i| - \nu$; otherwise, set $d_X(z_i, z_j) = |j - i|$. 
\end{enumerate}
\end{itemize}
In other words, other than $z_i$ and $z_j$ potentially deviates from the distance given by their coordinates, the distance between all  other pairs of points are consistent with their coordinates. It is easy to verify that $(X, d_X)$ as described above satisfies triangle inequality and is a valid metric. 
Let $\oset^*$ denote the optimal set of outliers so that $d_X$ restricted to $X \setminus \oset^*$ can be embedded in $\mathbb{R}^2$ isometrically. 
Let $C^*$ denote the optimal vertex cover for graph $G = (V, E)$. 
We now show $|\oset^*| = |C^*|$. 

Specifically, first note that given $C^*$, we can construct the set $\oset$ containing all $y_i$ such that $v_i \in C^*$. It is easy to see that the coordinates for the remaining points $X \setminus \oset$ are consistent with their pairwise distances. Hence $|\oset^*| \le |\oset| = |C^*|$. 
We next show that $|C^*| \le |\oset^*|$. 

Consider $\oset^*$. We construct the set $\hat{V} = \{ v_i \mid x_i \in \oset^* $ or $y_i \in \oset^*$ or $z_i \in \oset^* \}$. 
We now argue that $\hat V$ is a vertex cover for the graph $G = (V, E)$. 
Assume that this is not true, and that there exists an edge $(v_i, v_j) \in E$ such that $v_i \notin \hat V$ and $v_j \notin \hat V$. 
This means that the 6 points in $\{ x_i, y_i, z_i, x_j, y_j, z_j\}$ must all belong to $X \setminus \oset^*$. However, we argue that their pairwise distances are not consistent.
In particular, consider the 3 points $x_i, y_i, x_j$. They are isometrically embeddable into $\mathbb{R}^2$ (but not into $\mathbb{R}^1$). Hence once their coordinates are fixed, there is only a unique position possible for any other point $p$ given its distance to these three points. In other words, we can uniquely embed $z_i$, as well as $z_j$ into $\mathbb{R}^2$ w.r.t.~the current coordinates of $\{x_i, y_i, x_j \}$ (which form a valid isometry consistent with $d_X$). 
However, the Euclidean distance between $z_i$ and $z_j$ is not consistent with $d_X$. Hence the points $\{x_i, y_i, x_j, z_i, z_j\}$ are not isometrically embeddable into $\mathbb{R}^2$, which contradicts our assumption that all of them are in $X\setminus \oset^*$. 
It then follows that $\hat V$ has to be a vertex cover, and $|C^*| \le |\hat V| \le |\oset^*|$. 
This completed the proof for $d=2$.

We next discuss how the above construction can be generalized to the case $d>2$.
In the construction of $(X, d_X)$ we map each $v_i \in V$ to a sequence of $d+1$ points $x^{(0)}_i, \ldots, x^{(d)}_i$ in $X$ contained in the subspace $\mathbb{R}^{d-1}\times \{i\}$, where the first $d$ points form a canonical $(d-1)$-simplex,
while the last point $x^{(d)}_i$ is the centroid of this simplex.
The rest of the argument remains essentially identical so we omit the details.
\end{proof}

\section{Bi-criteria approximation for embedding into $\mathbb{R}^d$}
\label{appendix:sec:bicriteriaRd}
\newcommand{\mynormalized}	{normalized\xspace}

As before, let $d_E$ denote the Euclidean distance in $\mathbb{R}^d$. Given a point $x\in \mathbb{R}^d$ and a set $A \subset \mathbb{R}^d$, the distance from $x$ to $A$ is defined as $d(x, A) = \inf_{y\in A} d_E(x, y)$. 
\begin{definition}[Near-isometric $d$-embedding]
Given a metric space $\X = (X, \rho_X)$, we say that a map $\phi: X \to \mathbb{R}^d$ is a \emph{$\delta$-near-isometric embedding of $\X$ into $\mathbb{R}^d$} if for any $x, x'\in X$, we have that $| \rho(x, x') - d_E(\phi(x), \phi(x')) | \le \delta$. 
In this case, we also say that $\X$ admits a $\delta$-near-isometric $d$-embedding, and $(\phi(\X), d_E)$ is a $\delta$-approximation of $\X$. 
\end{definition}
We note that the parameter $\delta$-above is not normalized w.r.t.~the diameter of the set; this is for the convenience of the latter argument in this section. 
Before we describe our algorithm, we need some more definitions (following the notations of \cite{ATV01}). 

\begin{definition}[Normalized sequence]
Given an ordered sequence of points $U = \{ u_0, u_1, \ldots, u_d \} \subset \mathbb{R}^d$, 
we say that $U$ is \emph{\mynormalized} if $u_0 = \origin = (0, \ldots, 0)$, 
and if $u_i \in \mathbb{R}^i_+$ for all $1 \le i \le d$. 

Given $U \subseteq A \subset \mathbb{R}^d$, if $f: A \to \mathbb{R}^d$ is a map such that $f(U)$ is \mynormalized{}, then we say that the map $f$ is \emph{\mynormalized{} at $U$}. 
\end{definition}

\begin{definition}[Maximal sequence]
Let $k\leq d$.
Given an ordered sequence of points $U = \{ u_0, \ldots, u_k \}$ from a set $A \subset \mathbb{R}^d$, we say that $U$ is \emph{maximal} in $A$ if: (i) $d_E(u_0, u_1) = \diam(A)$, and (ii) for any $2 \le i \le k$, the distance $d(u_i, \mathrm{aff}\{u_0, \ldots, u_{i-1} \})$ is maximal in $A$. 
Here, $\mathrm{aff} B$ is the affine subspace of $\mathbb{R}^d$ spanned by points in set $B \subseteq \mathbb{R}^d$. 
\end{definition}

\paragraph{Description of the algorithm.} 
Assume that the input $n$-point metric space $\X = (X, \rho)$ admits a $(\eps, k^*)$-\noembedding{} into $\mathbb{R}^d$. 
We now present an algorithm, described in Algorithm \ref{alg:bicriteriaRd}, to find an approximation of such an embedding. 
The high level structure of this algorithm parallels that of Algorithm \ref{alg:outlierRd}, which intuitively is an algorithm for the special case when $\eps = 0$. 
However, the technical details here are more involved so as to tackle several issues caused by the near-isometric embedding. 
In this algorithm, $C_d$ is a positive constant depending only on the dimension $d$. 

\begin{algorithm}[htbp]
\caption{Bi-criteria approximation for embedding into $\mathbb{R}^d$. \label{alg:bicriteriaRd}}
\begin{algorithmic}
\Require An $n$-point metric space $(X, \rho)$, a parameter $0 \le \eps <1$, dimension $d$ 
\Ensure A set of outliers $\widehat{\oset} \subset X$ \\
~~~~~~~~~~~~Initialize the set of candidate outlier sets $\C$ to be empty. 
\State {\sf (Step-0)}~~For each subset of $d+1$ distinct points $Y_d = \{y_0, \ldots, y_d \} \subset X$, perform the following: \\
~~~~~~~~~~~~Initialize sets $W$ and $\oset$ to be the empty set. 
\begin{description}\denselist
\item ~~~~{\sf (Step-1)}~~Compute an \mynormalized{} sequence $P_d = \{p_0, p_1, \ldots, p_d \} \subset \mathbb{R}^d$ such that $(P_d, d_E)$ is an $(\eps \diam(X))$-approximation of $(Y_d, \rho)$. 
If such a $P_d$ does not exist, return to (Step-0). 
Otherwise, compute the smallest dimension $d' \le d$ for which such a $P_d$ exists. 
\item ~~~~{\sf (Step-2)}~~For each remaining point $x \in X \setminus Y_d$, compute $p_x \in \mathbb{R}^{d'}$ satisfying that $|\rho(x, y_i) - d_E(p_x, p_i)| \le C_d \sqrt{\eps}\diam({X}) + \eps \diam(X)$ for any $i\in [0, d]$. 
If such a $p_x$ can be found, insert the pair $(x, p_x)$ into the set $W$; otherwise, insert $x$ to the outlier set $\oset$.
\item ~~~~{\sf (Step-3)}~~Construct the graph $G = (W, E)$ where, for any two $w = (z, p), w' = (z', p') \in W$, $(w, w') \in E$ if and only if $| \rho(z, z') - d_E(p, p') |> C_d \sqrt{\eps} \diam(X) + \eps\diam(X)$. Compute a $2$-approximation $W' \subset W$ of the vertex cover of $G$. Set $\oset = \oset \cup \{ z \mid (z, p) \in W' \}$, and add the set $\oset$ to the collection of candiate outlier sets $\C$.
\end{description}
{\sf (Step-4)} Let $\widehat \oset$ be the set from $\C$ with smallest cardinality. We return $\widehat \oset$ as the outlier set. 
\end{algorithmic}
\end{algorithm}

For the time being, for the sake of simplicity of presenting the main ideas, let us assume that we can implement {\sf (Step-1,2)} (i.e, the computation of $P_d$ and $p_x$)) in polynomial time. We will address how to modify Algorithm \ref{alg:bicriteriaRd} to achieve this later in this section. 
Our main result is the following: 
\begin{theorem}\label{thm:bicriteriaRd}
Given an $n$-point metric space $\X = (X, \rho)$, if $\X$ admits a $(\eps, k^*)$-\noembedding{} in $\mathbb{R}^d$, then Algorithm \ref{alg:bicriteriaRd} returns an outlier set $\widehat{\oset}$ that witnesses an $(O(\sqrt{\eps}), |\widehat{\oset}|)$-\noembedding{} of $\X$ in $\mathbb{R}^d$ with $|\widehat{\oset}| \le 2k^*$. Here $O(\sqrt{\eps})$ hides constants depending on the dimension $d$. 
\end{theorem}
\paragraph{Proof of Theorem \ref{thm:bicriteriaRd}.} 
Let $\widehat{Z} = X \setminus \widehat{\oset}$.
By {\sf (Step-2)} and {\sf (Step-3)}, for any two points $z, z' \in \widehat Z$, we computed $p_z$ and $p_{z'} \in \mathbb{R}^d$ in (Step-2) such that $|\rho(z, z') - d_E(p_z, p_{z'}) | \le C_d \sqrt{\eps} \diam(X) + \eps \diam(X)$. 
Hence  $(\widehat{Z}, \rho)$ indeed has an $O(\sqrt{\eps} \diam(X))$-near-isometric embedding in $\mathbb{R}^d$. 

It remains to show that $|\widehat{\oset}| \le 2k^*$. 
Let $\oset^*$ be an optimal set of outliers such that $|\oset^*| = k^*$ and $(X \setminus \oset^*, \rho)$ admits a $\eps \diam(X)$-near-isometric $d$-embedding, say $\phi: X\setminus \oset^* \to \mathbb{R}^d$. 
Set $Z^* = X \setminus \oset^*$. 
For simplicity, we assume that there does not exist $d' < d$ such that $(X \setminus \oset^*, \rho)$ admits a $\eps \diam(X)$-near-isometric $d'$-embedding; the case where such $d' < d$ exists can be handled similarly. 

Let $Y^* = \{y_0, \ldots, y_d \}$ be the ordered sequence such that $\phi(Y^*)$ is maximal in $\phi(Z^*)$. We assume w.l.o.g.~that $\phi(Z^*)$ is also \mynormalized{} in $\mathbb{R}^d$; if not, there exists an isometry $T: \mathbb{R}^d \to \mathbb{R}^d$ so that $T(\phi(Y^*))$ is  \mynormalized{}, and we can simply take $\phi$ as $T \circ \phi$. 

Consider the time when our algorithm enumerates $Y_d = Y^*$ in {\sf (Step-0)}. 
In the subsequent {\sf (Step-1)}, the algorithm will be able to find an \mynormalized{} sequence $P_d$ as claimed, since $\phi(Y_d)$ satisfies the requirements. Following our earlier assumption on $Z^*$, $d' = d$ for $P_d$. 
In {\sf (Step-2)}, for a point $x\in X \setminus Z^*$, we claim that there exists $p_x \in \mathbb{R}^d$ such that $(P_d \cup \{ p_x \}, d_E)$ is an $2\eps \diam_X$-approximation of $(Y_d \cup \{x\}, \rho)$. 
To prove this, we need to use the following result proved in \cite{ATV01}. 
\begin{proposition}[Section 2.6, \cite{ATV01}]\label{prop:ATV}
For any fixed $\tilde d > 1$, let $A = \{\origin, u_1, \ldots, u_{\tilde d}, x \} \subset \mathbb{R}^{\tilde d}$ be such that the $\tilde d$-sequence $U = (\origin, u_1, \ldots, u_{\tilde d})$ is \mynormalized{} and maximal in $A$. Let $f: A \to \mathbb{R}^{\tilde d}$ be a map such that $(f(A), d_E)$ is an $\delta \diam(A)$-approximation of $(A, d_E)$, and $f$ is \mynormalized{} at $U$. Then there exists a constant $c_{\tilde d}$ depending only on the dimension $\tilde d$ such that $d_E(x, f(x)) \le c_{\tilde d} \sqrt{\delta} \diam(A)$. 
\end{proposition}

Consider the image $Q^* = \phi(Z^*)$ of an optimal embedding  $\phi: Z^* = X \setminus \oset^* \to \mathbb{R}^d$. 
Denote by $Q_d = \{ q_0 = \phi(y_0), q_1 = \phi(y_1), \ldots, q_d = \phi(y_d) \}$ and $g: Q_d \to P_d \subset \mathbb{R}^d$ the map that sends $q_i$ to $p_i$, for each $i\in [0, d]$. 
We also use $h: Y^* \to P_d$ to denote the $\eps \diam(X)$-near-isometric embedding computed in {\sf (Step-1)}. 

For any $x\in Z^*$, by the same argument as for the existence of $P_d$, we can show that there must exist a $\eps \diam(X)$-near-isometric embedding $h_x: Y_d \cup \{x\} \to \mathbb{R}^d$ of $(Y_d \cup \{x\}, \rho)$ such that the ordered sequence $\{ h_x(y_0), \ldots, h_x(y_d) \}$ is \mynormalized{}. Furthermore, Proposition \ref{prop:ATV} implies the following: 
\begin{claim}\label{claim:hx}
For any $x\in Z^*$ and any $h_x$ as described above, $d_E(\phi(x), h_x(x)) \le 2 c_d \sqrt{\eps} \diam(X)$. 
\end{claim}
\begin{proof}
For simplicity of presentation, set $y_{d+1} = x$, and consider the set $A = \{ q_0, \ldots, q_d, q_{d+1} = \phi(x) \} \subset Q^*$, and the map $g_x: A \to \mathbb{R}^d$ defined as $g_x(q_i) = h_x(y_i)$ for any $i\in [0, d+1]$. 
Since both $\phi$ and $h_x$ are $\eps \diam(X)$-near-isometric maps, we have by triangle inequality that: for any $i, j \in [0, d+1]$, 
$$|d_E(q_i, q_j) - d_E(h_x(y_i), h_x(y_j)| \le |d_E(q_i, q_j) - \rho(y_i, y_j)| + |\rho(y_i, y_j)  - d_E(h_x(y_i), h_x(y_j))| \le 2\eps \diam(X). $$
In other words, this means that the map $g_x$ is a $2\eps\diam(X)$-nearisometry for $(A, d_E)$. Setting $\delta = 2\eps \diam(X)/\diam(A)$, it then follows from Proposition \ref{prop:ATV} that 
$$d_E(q_{d+1}=\phi(x), h_x(x)) \le c_d \sqrt{\delta} \diam(A) = c_d \sqrt{ 2\eps \diam(X) \diam(A)} \le 2c_d \sqrt{\eps} \diam(X). $$ 
The last inequality follows that since $\phi$ is $\eps \diam(X)$-nearisometry, we have that $\diam(A) \le \diam(Q^*) \le (1+\eps)\diam(Z^*) \le 2 \diam(X)$ for $\eps \le 1$. 
\end{proof}

On the other hand, note that the above statement is generic for any $x \in Z^*$ and any $\eps \diam(X)$-near-isometric map $h_x$ 
such that the ordered sequence $\{ h_x(y_0), \ldots, h_x(y_d) \}$ is \mynormalized{}. 
Hence we can choose $x  = y_i$ for any $i\in [0, d]$ and choose $h_x = h$ computed in (Step-1) $h: Y_d \to P_d$ that gives rise to $P_d$. 
It then follows from Claim \ref{claim:hx} that: 
\begin{equation}\label{eqn:Pd}
\text{for any } i\in [0, d], ~d_E(\phi(y_i), p_i = h(y_i)) \le 2 c_d \sqrt{\eps} \diam(X). 
\end{equation}
Finally for any $x \in X \setminus Y_d$ inspected in (Step-2), set $p_x = h_x(x)$ for any $\eps \diam(X)$-near-isometry $h_x: Y^* \cup \{x\} \to \mathbb{R}^d$ (which must exist as we argued earlier). 
Combining Claim \ref{claim:hx}, Eqn \eqref{eqn:Pd}, the triangle inequality and the fact that $\phi$ is an $\eps \diam(X)$-nearisometry, we thus have: 
\begin{align}
|\rho(y_i, x) - d_E(p_i, p_x)| &\le |\rho(y_i, x) - d_E(\phi(y_i), \phi(x)) | + |d_E(\phi(y_i), \phi(x)) - d_E(p_i, p_x)| \nonumber \\
&\le \eps \diam(X) + d_E(\phi(y_i), p_i) + d_E(\phi(x), p_x) \le \eps \diam(X) + 4c_d \sqrt{\eps}\diam(X). \label{eqn:Pdx}
\end{align}
Setting $C_d = 4c_d$, we obtain the following.
\begin{corollary}\label{cor:step2}
For any $x\in Z^*$, there exists $p_x$ satisfying the requirements in {\sf (Step-2)}. 
\end{corollary}
By a similar argument as in Eqn (\ref{eqn:Pdx}), we can also obtain: 
\begin{corollary}\label{lem:pairconsistency}
For any $z, z' \in Z^*$, let $p_z$ and $p_{z'}$ be their corresponding points in $\mathbb{R}^d$ computed in {\sf (Step-2)}. Then $| \rho(z, z') - d_E(p_z, p_{z'}) | \le C_d \sqrt{\eps} \diam(X) + \eps\diam(X)$. 
\end{corollary}

\begin{lemma}
$|\hat{\oset}| \le 2 |\oset^*|$.
\end{lemma}
\begin{proof}
First, as argued earlier above Corollary \ref{cor:step2}, for any $x \in Z^*$, {\sf (Step-2)} will be able to compute the required $p_x \in \mathbb{R}^d$. Let $\oset_1$ denote the set $O$ at the end of {\sf (Step-2)}. It then follows that $\oset_1 \subset \oset^*$. 
Furthermore, by Corollary \ref{lem:pairconsistency}, for any edge $( (z, p), (z', p') )\in E$ in the graph constructed in Step-3, at least one of $z$ and $z'$ is necessarily in $\oset^*$. 
Hence $|\oset^* \setminus \oset_1|$ is at least the size of the vertex cover of $G$. Since we compute a $2$-approximation of the vertex cover, it follows that the set $\oset$ at the end of Step-3 satisfies that $|\oset| \le 2|\oset^*|$. 
This proves the claim as $|\hat{\oset}| \le |\oset|$. 
\end{proof}

This completes the proof of Theorem \ref{thm:bicriteriaRd}.

\paragraph{Computational issues.}
The computation of $P_d$ and $p_x$ in {\sf (Step-1)} and {\sf (Step-2)} require solving a system of $O(d)$ quadratic equations and inequalities, and the solution would require infinite precision. This can be addressed using real algebraic geometry theory for semi-algebraic sets which can output the solutions in arbitrary precision in time exponential in $O(d)$. 
We can also use the following ``griding" strategy: 
After {\sf (Step-1)}, we discretize the $d$-dimensional cube $[-3\diam(X), 3\diam(X)]^d$ where the grid length is  $\tau = \frac{\eps}{2d^2} \diam(X)$. 
In {\sf (Step-1)} and {\sf (Step-2)}, we look for solutions for $P_d$ and $p_x$ where coordinates fall on grid points. This requires us to relax the distance distortion bound by another factor of $\eps \diam(X)$. However it will not change the final statement in Theorem \ref{thm:bicriteriaRd}. 
The time to compute each near-isometry embedding is $(\frac{d}{\eps})^{O(d)}$. 
Overall, following a similar analysis as for Algorithm \ref{alg:outlierRd}, we conclude with Theorem \ref{thm:bicriteriaRdFinal}.

%

\end{document}